\documentclass[10pt,letterpaper,conference]{ieeeconf}
\usepackage[latin9]{inputenc}
\usepackage{color}
\usepackage{float}
\usepackage{mathtools}
\usepackage{amsmath}
\usepackage{amssymb}
\usepackage{tikz,algorithmic}
\usetikzlibrary{matrix,decorations.pathreplacing, calc, positioning,fit}
\usepackage[linesnumbered,ruled,vlined]{algorithm2e}
\usepackage{graphics}
\usepackage{epstopdf}
\usepackage{amsmath,amssymb,amsfonts,mathtools,enumerate,graphicx,tikz,pgfplots,mathdots,lipsum,cuted}
\usepackage{pgfplots}
\usepackage{varwidth}
\usepackage{stfloats,lipsum} 
\usepackage{adjustbox}
\usepackage{arydshln}
\usepackage{bm}

\IEEEoverridecommandlockouts
\floatstyle{ruled}

\newtheorem{assumption}{Assumption}
\newtheorem{theorem}{Theorem}

\newtheorem{lemma}{Lemma}

\newtheorem{prob}{Problem}

\usepackage{cite}

\title{Learning from similar systems and online data-driven LQR using iterative
randomised data compression}
\author{Vatsal Kedia, Sneha Susan George and Debraj Chakraborty
\thanks{The authors are with the Department of Electrical Engineering, Indian Institute of Technology Bombay, Mumbai, Maharashtra, India. 
		{\tt\small \{vatsalkedia,snehasg,dc\}@ee.iitb.ac.in}}
  }

\begin{document}

\maketitle
	\thispagestyle{empty}
	\pagestyle{empty}
\begin{abstract}
The problem of data driven recursive computation of receding horizon LQR control
through a randomized combination of online/current and historical/recorded data, is considered. It is assumed that large amounts
of historical input-output data from a system, which is similar but
not identical to the current system under consideration, is available.
This (possibly large) data set is compressed through a novel randomized
subspace algorithm to directly synthesize an initial solution
of the standard LQR problem, which however is sub-optimal due to the inaccuracy of the historical model. The first instance of this input is used to
actuate the current system and the corresponding instantaneous output is used
to iteratively re-solve the LQR problem through a computationally inexpensive randomized rank-one update of the old compressed data. The first instance of the re-computed input is applied to the system at the next instant, output recorded and the entire procedure is repeated at each subsequent instant. As more current data becomes available, the algorithm learns automatically from the new data while simultaneously controlling the system in near optimal manner. The proposed algorithm is computationally inexpensive due to the initial and repeated compression of old and newly available data. Moreover, the simultaneous learning and control makes this algorithm
particularly suited for adapting to unknown, poorly modeled and time
varying systems without any explicit exploration stage. Simulations
demonstrate the effectiveness of the proposed algorithm vs popular
exploration/exploitation approaches to LQR control.
\end{abstract}

\section{INTRODUCTION}

Conventional system identification involves collecting input-output (I/O)
data through controlled experiments and thereafter estimating model
parameters from this data using any one among a rich
variety of available algorithms (e.g. see \cite{Ljung_book}, \cite{BDMoor_SID_book}
and references therein). Any model based control algorithm can only
be started after this required phase of parametric identification
is completed successfully. However, the model thus formed is traditionally
linear time invariant while the actual system might be more complex,
non-linear, time varying or simply might have changed marginally
due to wear and tear from the time it was last identified. In such a situation, re-identification
is the recommended procedure, which however implies extra cost, effort
and postponement and/or stoppage of the actual controlled operation. Moreover conventional re-identification of model parameters is computationally expensive and may only be done
after sufficient new data is collected. All these issues predicate
that even when re-identification is absolutely essential, that it
is done at a rate which is several orders of magnitude slower than
the time constants of the controlled system. This leads to accumulation of errors
in the control sequences, sub-optimal performances and in extreme cases, might lead to instability of the closed loop. To address
the above issues, we propose a compressed data driven iterative technique
to directly compute the optimal control at each instant by seamlessly
combining past recorded data with the "online" I/O data as and when it becomes available
at the current instant of time. The inbuilt data compression
combined with the iterative nature of the control sequence computation
makes our technique computationally inexpensive enough to be implemented
in real time on current control hardware.

A recently proposed paradigm at the boundary of system identification
and transfer learning theory \cite{pan2009survey} involves the estimation of system models using
a combination of historical/recorded data from a similar/auxiliary
model as well as the current model \cite{xin2023learning,9867413,lau2023multi,toso2023learning,zheng2020non,tu2022learning}. This method effectively
addresses the common problem of shortage/unavailability of data and/or
difficulties in I/O data collection for the current, to-be-controlled
system. The primary advantage of such a procedure is the increased
robustness of the parameter estimation process to noise due to increased
data size, while increased errors might result due to the mismatch between
the auxiliary and the current system. Estimates of finite sample errors in system identification problems in the context of LQR have been studied recently in \cite{dean2020sample}. However the only available techniques to handle such errors remain in the domain of robust and/or adaptive control \cite{zhou1998essentials,aastrom2013adaptive}. 

Recursive corrections to the model with streaming data and simultaneous control updates has also been investigated \cite{kim2008adaptive,tanaskovic2014adaptive,lorenzen2019robust,bujarbaruah2020adaptive}. Recently, machine learning methods such as exploration and exploitation to optimally combine the identification and control stages in the context of LQR control has been studied rigorously \cite{abbasi2011regret,abeille2017thompson,dean2018regret}.  However, these methods are predominantly model based and computationally expensive.

On the other hand, another recent paradigm in control (and model predictive
control in particular) is the use of I/O data directly to
design a controller without explicitly identifying the underlying
model \cite{willems2005note,de2019formulas,allibhoy2020data,berberich2020data,dai2023data}. These techniques have been widely studied for the implementation
of predictive control techniques such as receding horizon, iterative or infinite horizon
LQR \cite{rotulo2020data,da2018data}.

While both transfer learning and data driven controller synthesis are especially well suited for complex and time varying environments, their implementations with limited
computational resources, such as in robotics, cyber physical systems, networked control systems, automotive control etc, where
small onboard computers are common, are difficult due to the requirement
of relatively large computing power. Hence in this paper we propose a computationally efficient amalgamation of the above two ideas. We use randomized (Gaussian) data compression matrices to efficiently learn from ``large" historical I/O
databases \cite{kedia2023randomized,kedia2023fast} and simultaneously update this compressed data with online
or current data as they become available. The updating of the data
is done using a pure iteration with negligible computational cost.
Thereafter the updated data is used to directly compute a receding
horizon LQR control sequence without identifying any model. It turns
out that the LQR control sequence computation can also be implemented
as a pseudo-iteration based on the underlying iterative computation
of the compressed data matrices and a rank one update of an intermediate QR factorization \cite{golub2013matrix}. Due to the compression of the historical
(long) database to conveniently small sizes, the entire operation
is computationally extremely efficient. Moreover the iterative nature
of all the major computations makes it fast enough for real time implementations
in situations with limited computing resources.

The proposed algorithm is compared with a basic exploration-exploitation scheme, where the exploration stage consists of a minimum variance white noise excitation to identify the system. This experimental phase is conducted for the minimum duration required by popular subspace identification algorithms \cite{BDMoor_SID_book}. The identified model is used to construct a standard model based receding horizon LQR controller, whose performance cost is compared with that of the proposed algorithm. It turns out that the cost saved by the proposed method due to its inbuilt knowledge of the "similar" system, especially during the initial stages of controlled operation, enables it to vastly outperform such online identification approaches.  At the same time, except for the initial offline compression of similar system data, all other  online updates are extremely fast.

\section{Preliminaries and Problem Formulation}

Consider the following discrete LTI system 
\begin{equation}
\begin{aligned} & x(t+1)=Ax(t)+Bu(t)+e(t)\\
 & y(t)=Cx(t)
\end{aligned}
\label{eq_ssmodel}
\end{equation}

Here $e(t)\in\mathbb{R}^{n}$ is the process
noise and is assumed to be a white noise sequence with zero mean
and finite covariance i.e. $\mathbb{E}\{e(t_{1})e^{T}(t_{2})\}=\eta\delta_{t_{1}t_{2}}$
for all time instants $t_{1}$ and $t_{2}$, where, $\eta\in\mathbb{R}^{n\times n}>0$
and $\delta$ is the Kronecker delta function. The system parameters
$\{A,B,C\}$ are of appropriate dimensions: $A\in\mathbb{R}^{n\times n}$,
$B\in\mathbb{R}^{n\times m}$ and $C\in\mathbb{R}^{p\times n}$. We assume that the state information can be measured directly \cite{rotulo2020data} (i.e. $C = I_{n}$ and $p=n$). In the next subsection we briefly review the batch and receding horizon LQR formulations based on \cite{borrelli2017predictive}.

\subsection{Linear Quadratic Regulator}

Let us consider the deterministic version of the model \eqref{eq_ssmodel} for simplicity. We define the following quadratic cost function over a finite horizon
$k_{p}$ as: 
\begin{equation}\label{eq_costfunction}
J(x_{0},\mathcal{U}_{0})=x(k_{p})^{T}Px(k_{p})+\sum_{i=0}^{k_{p}-1}x(i)^{T}Qx(i)+u(i)^{T}Ru(i)
\end{equation}
where $Q=Q^{T}\succcurlyeq0\in\mathbb{R}^{n\times n}$, $P=P^{T}\succcurlyeq0\in\mathbb{R}^{n\times n}$
and $R\succ0\in\mathbb{R}^{m\times m}$, $x(0)=x_0$ and $\mathcal{U}_{0}=[u(0)\; u(1)\;...u(k_p-1)]^T$ .

Consider the finite time optimal control problem: 
\begin{equation}
\begin{aligned}\label{eq_optimal_cost}
J_{0}^{*}(x_{0}) & =\min_{\mathcal{U}_{0}}J(x_{0},\mathcal{U}_{0})\\
 & \text{s.t. }x(t+1)=Ax(t)+Bu(t)\hspace{0.2cm}\forall t=0,1,\hdots,k_{p}-1
\end{aligned}
\end{equation}

Next, we write the system evolution in terms of inputs and initial condition.

\begin{equation}\label{eq_LQRmatrixform}
\begin{aligned}\underbrace{\begin{bmatrix}x(0)\\
x(1)\\
\vdots\\
x(k_{p})
\end{bmatrix}}_{\mathcal{X}_{0}}=\underbrace{\begin{bmatrix}I\\
A\\
\vdots\\
A^{k_{p}}
\end{bmatrix}}_{\mathcal{S}_{*}^{x}}x(0)+\underbrace{\begin{bmatrix}0 & 0 & \hdots & 0\\
B & 0 & \ddots & 0\\
AB & \ddots & \ddots & \vdots\\
\vdots & \ddots & \ddots & 0\\
A^{k_{p}-1}B & \hdots & \hdots & B
\end{bmatrix}}_{\mathcal{S}_{*}^{u}}\underbrace{\begin{bmatrix}u(0)\\
u(1)\\
\vdots\\
u(k_{p}-1)
\end{bmatrix}}_{\mathcal{U}_{0}}\end{aligned}\end{equation}

where, $\mathcal{X}_{0}\in\mathbb{R}^{n(k_{p}+1)}$, $\mathcal{S}_{*}^{x}\in\mathbb{R}^{n(k_{p}+1)\times n}$,
$\mathcal{S}_{*}^{u}\in\mathbb{R}^{n(k_{p}+1)\times k_{p}m}$ and $\mathcal{U}_{0}\in\mathbb{R}^{k_{p}m}$.
The cost function $J(x_{0},\mathcal{U}_{0})$ can be rewritten as:
\begin{equation}\label{eq_optimalcostfunction_matrix}
J(x(0),\mathcal{U}_{0})=\mathcal{X}_{0}^{T}\bar{Q}\mathcal{X}_{0}+\mathcal{U}_{0}^{T}\bar{R}\mathcal{U}_{0}
\end{equation}
where $\bar{Q}=blockdiag\{Q,Q,\hdots,Q,P\}\in\mathbb{R}^{n(k_{p}+1)\times n(k_{p}+1)}$
and $\bar{R}=blockdiag\{R,R,\hdots,R\}\in\mathbb{R}^{k_{p}m\times k_{p}m}$.
The optimal input sequence is given by: 
\begin{equation}
\label{opt_input}
\mathcal{U}_{0}^{*}(x(0))=-(\mathcal{S}_{*}^{u^{T}}\bar{Q}\mathcal{S}_{*}^{u}+\bar{R})^{-1}\mathcal{S}_{*}^{u^{T}}\bar{Q}\mathcal{S}_{*}^{x}x(0)
\end{equation}
The optimal cost $J_{0}^{*}(x(0))$ can be easily calculated by substituting
the above result in \eqref{eq_optimalcostfunction_matrix}.

As opposed to the open loop formulation described above,  a similar but closed loop version of  \eqref{opt_input}  is used in the receding horizon framework. Here, an open-loop optimal control problem is
solved over a finite horizon of $k_{p}$ steps at each sampling instant. Only the first instance of the
optimal input sequence $(\mathcal{U}^{*})$ is applied to the process. At
the next sampling instant, a new optimal control problem is solved
over a shifted horizon of $k_{p}$-steps based on new measurements of the current states. The
optimal input sequence at $i^{th}$- sampling instant is given by:
\begin{equation}
\mathcal{U}_{i}^{*}(x(i))=-(\mathcal{S}_{*}^{u^{T}}\bar{Q}\mathcal{S}_{*}^{u}+\bar{R})^{-1}\mathcal{S}_{*}^{u^{T}}\bar{Q}\mathcal{S}_{*}^{x}x(i)\label{eq_optimalu}
\end{equation}

The above optimal input sequence works perfectly under the scenario
where we know the system dynamics, and hence   $\mathcal{S}_{*}^{x}$ and $\mathcal{S}_{*}^{u}$, exactly. 

\subsection{Problem Formulation}

Let the dynamics of a similar system be given by 
\begin{equation}
\begin{aligned} & x(t+1)=\tilde{A}x(t)+\tilde{B}u(t) + e(t)\\
 & y(t)=x(t)
\end{aligned}
\label{ss_similarmodel1}
\end{equation}
where, $\tilde{A}=A+\Delta A$, $\tilde{B}=B+\Delta B$, the pair $(A, B)$ is our actual system while the pair $(\tilde{A}, \tilde{B})$ represents the similar system with ($\Delta A, \Delta B$) being the unknown perturbation matrices. The other variables in the equation are identical to those defined in the original system \eqref{eq_ssmodel}. We assume that we have access to the I/O data generated previously by the above system, either through separate identification experiments or through regular controlled use. The generated I/O data set is denoted by $\{\tilde{u}(i),\tilde{y}(i)\}_{i=0}^{N_{t}-1}$.
\begin{prob}
\label{section:problem_statement} Design a data-driven ``efficient"
method to iteratively compute the receding LQR solution given in \eqref{eq_optimalu} by leveraging the data set $\{\tilde{u}(i),\tilde{y}(i)\}_{i=0}^{N_{t}-1}$ recorded from the similar system \eqref{ss_similarmodel1} and online data available at each iteration on application of \eqref{eq_optimalu}  on the 
actual system \eqref{eq_ssmodel}.

\end{prob}
Some assumptions standard in the system identification literature are listed below. 
\begin{assumption}
    \label{assumption:system} \cite{Ljung_book,BDMoor_SID_book} The following are assumed:
\begin{enumerate}
\item The input $u(t)$ is persistently exciting. 
\item The input $u(t)$ is uncorrelated with $e(t)$. 
\item The signal to noise ratio is high.
\end{enumerate}
\end{assumption}
\section{Compressed Data Driven LQR synthesis}
In this section we present the two basic building blocks for the proposed method, namely the computation of $\mathcal{S}^{x}$ and $\mathcal{S}^{u}$ directly from I/O data and the same computation based on compressed I/O data.
\subsection{Data driven learning of $\mathcal{S}^{x}$ and $\mathcal{S}^{u}$:
Subspace approach}

Given input-output time-series data sequence $\{u(i),y(i)\}$ $\forall i\in\{0,1,\hdots,N_{t}-1\}$.
First, a horizon $k = k_{p}+1$ is chosen where $k_{p}$ is the prediction horizon for the LQR problem and the block size $N:=N_{t}-2k+2$
is defined. Let us denote a block Hankel matrix based on input sequence
$\{u(i)\}$ as follows: 
\begin{equation}
	U_{i|i+k-1} :=\begin{bsmallmatrix}u(i) & u(i+1) & \hdots & u(i+N-1)\\
		u(i+1) & u(i+2) & \hdots & u(i+N)\\
		\vdots & \vdots & \ddots & \vdots\\
		u(i+k-1) & u(i+k) & \hdots & u(i+k+N-2)
	\end{bsmallmatrix}\label{hankel_matrix}
\end{equation}
Now, the past input block Hankel matrix is defined as $U_{p}:=U_{0|k-1}\in\mathbb{R}^{km\times N}$
by substituting $i=0$ above, while the corresponding future input
matrix is defined as $U_{f}:=U_{k|2k-1}\in\mathbb{R}^{km\times N}$
(substituting $i=k$). Similarly we define the output block Hankel
matrices $Y_{p},Y_{f}\in\mathbb{R}^{kn\times N}$ using the
past/future output data $\{y(i)\}$. Note that the nomenclature "past" and "future" is merely a mnemonic used in the subspace identification literature and has little to do with actual time instants. Although no recordings of noise are assumed to be available, for the sake of notational convenience, similar matrices are also defined for the corresponding past and future innovations
processes: $E_{p},E_{f}\in\mathbb{R}^{kn\times N}$. The past input and output data
is combined into $W_{p}:=\begin{bmatrix}U_{p}^{T} & Y_{p}^{T}\end{bmatrix}^{T}\in\mathbb{R}^{k(m+n)\times N}$. Let $X_{f}\in\mathbb{R}^{n\times N}$ denote the future state sequence
defined as $X_{f}:=\begin{bmatrix}x(k) & x(k+1) & \hdots & x(k+N-1)\end{bmatrix}\in\mathbb{R}^{n\times N}$. We further denote, $\Phi_{k}\in\mathbb{R}^{kn\times kn}$
as noise impulse response Toeplitz matrix as shown below. 
 \begin{equation*}
      \Phi_{k}=\begin{bmatrix}0 & 0 & \hdots & 0\\
I & 0 & \ddots & \vdots\\
\vdots & \ddots & \ddots & 0\\
A^{k-2} & \hdots & I & 0
\end{bmatrix}
 \end{equation*}
Recursively using \eqref{eq_ssmodel} and data matrices defined above
we get,
\begin{equation}
\begin{aligned} & Y_{f}=\mathcal{S}^{x}X_{f}+\begin{bmatrix}\mathcal{S}^{u} & 0_{m}\end{bmatrix}U_{f}+\Phi_{k}E_{f}.\end{aligned}
\label{Yfuture}
\end{equation}
where $0_{m}$ is zero matrix of size $km \times m$ and $\begin{bmatrix}\mathcal{S}^{u} & 0_{m}\end{bmatrix} \in \mathbb{R}^{kn \times km}$. Let $\bar{A}:=A-C$, $\bar{B}:=B$, $\Upsilon_{k}:=\begin{bmatrix}\bar{A}^{k-1}\bar{B} & \bar{A}^{k-2}\bar{B} & \hdots & \bar{B}\end{bmatrix}\in\mathbb{R}^{n\times km}$
be the modified reversed extended controllability matrix and $\Upsilon_{k}^{e}:=\begin{bmatrix}\bar{A}^{k-1} & \bar{A}^{k-2} & \hdots & I\end{bmatrix}\in\mathbb{R}^{n\times kn}$
be the modified reversed extended stochastic controllability matrix. Then, under the assumptions listed above and for large prediction
horizons $k$ it can be shown \cite{LjungSID} that $X_{f}=L_{p}W_{p}$
for $L_{p}:=\begin{bmatrix}{\Upsilon}_{k} & {\Upsilon}_{k}^{e}\end{bmatrix}\in\mathbb{R}^{n\times k(m+n)}$. Thereby  \eqref{Yfuture} reduces to 
\begin{equation}
Y_{f}=\mathcal{S}^{x}L_{p}W_{p}+\begin{bmatrix}\mathcal{S}^{u} & 0\end{bmatrix}U_{f}+\Phi_{k}E_{f}.\label{Yfuturesto}
\end{equation}
Under assumption \ref{assumption:system}, the projection $Y_{f}$ onto the joint span of $W_{p}$ and $U_{f}$ becomes:
	\begin{equation}\small
		\begin{aligned}Y_{f}/\begin{bmatrix}W_{p}\\
				U_{f}
			\end{bmatrix} =\mathcal{S}^{x}{L}_{p}W_{p}+\begin{bmatrix}\mathcal{S}^{u} & 0\end{bmatrix}U_{f}
		\end{aligned}
		\label{eq_Yfuturefinal}
	\end{equation}
Now $Y_{f}$ orthogonally projected onto the joint span of $W_{p}$
and $U_{f}$ can also be written as, 
\begin{equation}
\begin{aligned}Y_{f}/\begin{bmatrix}W_{p}\\
U_{f}
\end{bmatrix} & =Y_{f}/_{U_{f}}W_{p}+Y_{f}/_{W_{p}}U_{f}\\
 & =\underbrace{\bar{L}_{p}W_{p}}_{:=\zeta}+L_{U_{f}}U_{f}
\end{aligned}
\label{Yf_ortho}
\end{equation}
where $\zeta:=Y_{f}/_{U_{f}}W_{p}\in\mathbb{R}^{kp\times N}$ is known as the oblique projection of $Y_{f}$ onto $W_{p}$ along $U_{f}$. On comparing \eqref{eq_Yfuturefinal} and \eqref{Yf_ortho}
we get $\bar{L}_{p}=\mathcal{S}^{x}L_{p}$, $L_{U_{f}} = \begin{bmatrix}\mathcal{S}^{u} & 0\end{bmatrix}$ . Define $\zeta=\bar{L}_{p}W_{p}\in\mathbb{R}^{kp\times N}$.
An efficient way to calculate the oblique projection is by using the
$QR$ decomposition.
\subsubsection{QR step}

Perform LQ decomposition on $H:=\begin{bmatrix}U_{f}^{T} & W_{p}^{T} & Y_{f}^{T}\end{bmatrix}^{T}\in\mathbb{R}^{2k(m+n)\times N}$
to obtain the decomposition of $Y_{f}$ as shown in \eqref{Yfuturesto}.
\begin{equation}\label{eq_Hmatrix}
\begin{aligned}H & =\begin{bmatrix}U_{f}\\
W_{p}\\
Y_{f}
\end{bmatrix} & =\underbrace{\begin{bmatrix}R_{11} & 0 & 0\\
R_{21} & R_{22} & 0\\
R_{31} & R_{32} & 0
\end{bmatrix}}_{L}\begin{bmatrix}Q_{1}^{T}\\
Q_{2}^{T}\\
Q_{3}^{T}
\end{bmatrix}\end{aligned}
\end{equation}
From \eqref{eq_Hmatrix}, 
\begin{equation}
Y_{f}/\begin{bmatrix}W_{p}\\
U_{f}
\end{bmatrix}=R_{32}R_{22}^{\dagger}W_{p}+(R_{31}-R_{32}R_{22}^{\dagger}R_{21})R_{11}^{-1}U_{f}\label{eq_Yf_ortho_data}
\end{equation}
On comparing \eqref{eq_Yfuturefinal} and \eqref{eq_Yf_ortho_data}, 
\begin{equation}
\begin{aligned}  \bar{L}_{p}&=R_{32}R_{22}^{\dagger}\\
  \begin{bmatrix}\mathcal{S}^{u} & 0\end{bmatrix} &=(R_{31}-R_{32}R_{22}^{\dagger}R_{21})R_{11}^{-1}
\end{aligned}
\label{zetaPsik}
\end{equation}
Now it can be shown using \eqref{Yfuture}, \eqref{eq_Yfuturefinal} and \eqref{Yf_ortho} that 
\begin{equation}
\underbrace{\mathcal{S}^{x}X_{f}}_{theoretical}=\bar{L}_{p}W_{p}=\underbrace{R_{32}R_{22}^{\dagger}W_p}_{data}=:\zeta.\label{eqzeta}
\end{equation}
\subsubsection{SVD step}

Next we calculate the SVD of $\zeta$ as follows:
	\begin{equation}
		\begin{aligned}\zeta & =\begin{bmatrix}U_{1} & U_{2}\end{bmatrix}\begin{bmatrix}\Sigma_{1} & 0\\
				0 & \Sigma_{2}
			\end{bmatrix}\begin{bmatrix}V_{1}^{T}\\
				V_{2}^{T}
			\end{bmatrix}\\
			& =U_{1}\Sigma_{1}V_{1}^{T}+\underbrace{U_{2}\Sigma_{2}V_{2}^{T}}_{noise}\\
			& \approx U_{1}\Sigma_{1}V_{1}^{T}=\underbrace{U_{1}\Sigma_{1}^{1/2}T}_{{\mathcal{S}^{x}}}T^{-1}\Sigma_{1}^{1/2}V_{1}^{T}
		\end{aligned}
		\label{eqzetasvd}
	\end{equation}
where $T \in \mathbb{R}^{n \times n}$ is an arbitrary similarity transformation matrix. The second term is ignored assuming that the noise component is negligible as compared to the system contribution. Therefore, $\mathcal{S}^{x}$ and $\mathcal{S}^{u}$ can be computed directly from input-output data up to similarity transforms. Consequently, the optimal input $\mathcal{U}_{*}$ using \eqref{eq_optimalu} becomes:
 \begin{equation}
\mathcal{U}_{i}^{*}(x(i))=-(\mathcal{S}^{u^{T}}\bar{Q}\mathcal{S}^{u}+\bar{R})^{-1}\mathcal{S}^{u^{T}}\bar{Q}\mathcal{S}^{x}T^{-1}x(i)\label{eq_optimalu_transformed}
\end{equation}
In order to use the above feedback law, we need to compute $T$. It is easy to see that on comparing $\mathcal{S}^{x}_{*}$ (see \eqref{eq_LQRmatrixform}) and $\mathcal{S}^{x}$, the first $n$-rows of $\mathcal{S}^{x}$ forms $T$ i.e. $T = \mathcal{S}^{x}(1:n,:)$.

\subsection{Efficient learning: Randomized approach}\label{section:effiecient_learning_randomized}

We propose randomized data compression based computation of  $\mathcal{S}^{x}$ and  $\mathcal{S}^{u}$.

\subsubsection{Data Compression:}

\label{subsection:SDC} Recall that the matrix $H:=\begin{bmatrix}U_{f}^{T} & W_{p}^{T} & Y_{f}^{T}\end{bmatrix}^{T}\in\mathbb{R}^{2k(m+n)\times N}$
and define $N_{c}:=2k(m+n)+l$ where, $l>0$ is commonly known as
	the oversampling parameter \cite{halko2011_findstructurerandom}. Define $\mathcal{C}\in\mathbb{R}^{N\times N_{c}}$
	to be a random matrix whose elements are iid gaussian with $(\mathcal{C})_{ij}\sim\mathcal{N}(0,\frac{1}{N_{c}})$.
We further define, $\bar{U}_{f}:=U_{f}\mathcal{C}$, $\bar{U}_{p}:=U_{p}\mathcal{C}$,
$\bar{Y}_{p}:=Y_{p}\mathcal{C}$, $\bar{Y}_{f}:=Y_{f}\mathcal{C}$,
$\bar{W}_{p}:=W_{p}\mathcal{C}=\begin{bmatrix}\bar{U}_{p}^{T} & \bar{Y}_{p}^{T}\end{bmatrix}^{T}$
and $\bar{H}:=H\mathcal{C}=\begin{bmatrix}\bar{U}_{f}^{T} & \bar{W}_{p}^{T} & \bar{Y}_{f}^{T}\end{bmatrix}^{T}\in\mathbb{R}^{2k(m+n)\times N_{c}}$.

\subsubsection{LQ Decomposition}\label{subsection:QRstep}

In the uncompressed case, the first step to compute  $\mathcal{S}^{x}$ and  $\mathcal{S}^{u}$ is to perform QR on $H^{T}$ to obtain $\mathcal{S}^{x}$
and $\mathcal{S}^{u}$ (see \eqref{eq_Hmatrix}). Since, we need only the $R$-factor from the
QR step, hence instead, we propose to perform QR on the compressed
matrix $\bar{H}^{T}\in\mathbb{R}^{N_{c}\times 2k(m+n)}$. This
step reduces the QR computation cost significantly since $N_{c}<<N$. However for this method to work, we must
show that the projection can still be used to extract the
desired subspace even after data compression.

First, note that an equation similar to \eqref{eq_Yf_ortho_data} can
be obtained for the compressed case using QR decomposition of $\bar{H}^{T}$:
\begin{equation}\label{eq_Yfbar_ortho}
\bar{Y}_{f}/\begin{bmatrix}\bar{W}_{p}\\
\bar{U}_{f}
\end{bmatrix}=\bar{R}_{32}\bar{R}_{22}^{\dagger}\bar{W}_{p}+(\bar{R}_{31}-\bar{R}_{32}\bar{R}_{22}^{\dagger}\bar{R}_{21})\bar{R}_{11}^{-1}\bar{U}_{f}\
\end{equation}
where $\bar{(.)}$ denote the equivalent matrices for the compressed
case. 
Now, we right multiply \eqref{Yfuturesto} by $\mathcal{C}$ 
\begin{equation}
\bar{Y}_{f}=\bar{L}_{p}\bar{W}_{p}+\begin{bmatrix}
    \mathcal{S}^{u} & 0
\end{bmatrix}\bar{U}_{f}+\Phi_{k}\bar{E}_{f}\label{eqYfinalOmega}
\end{equation}
Note that, the $\bar{L}_{p}$ remains the same as in the uncompressed case
due to multiplication from the right by $\mathcal{C}$ on \eqref{Yfuturesto}.
Under assumption \ref{assumption:system}, the orthogonal projection of $\bar{Y}_{f}$ onto the joint span of
$\bar{W}_{p}$ and $\bar{U}_{f}$ is 
\begin{equation}
\begin{aligned}\label{Ybarfinal}\bar{Y}_{f}/\begin{bmatrix}\bar{W}_{p}\\
\bar{U}_{f}
\end{bmatrix}
 & \approx\bar{L}_{p}\bar{W}_{p}+\begin{bmatrix}
    \mathcal{S}^{u} & 0
\end{bmatrix}\bar{U}_{f}
\end{aligned}
\end{equation}
The above result ensures that we can use the projection to extract the
desired subspace under assumption \ref{assumption:system} (similar
to the uncompressed case). Recall that the oblique projection $\zeta:=Y_{f}/_{U_{f}}W_{p}\in\mathbb{R}^{kp\times N}$
and define the compressed version $\bar{\zeta}:=\bar{Y}_{f}/_{\bar{U}_{f}}\bar{W}_{p}\in\mathbb{R}^{kp\times N_{c}}$.
Then the following result can be proved along the lines of \cite{kedia2023randomized,kedia2023fast} and is not included here in the interest of space.
\begin{lemma}
\label{Lemma_obliqueproj} $\mathcal{R}(\bar{\zeta})=\mathcal{R}(\zeta)$
almost surely.
\end{lemma}

Therefore, various projections of $\bar{Y}_{f}$ onto $\bar{U}_{f}$
and $\bar{W}_{p}$ as shown in \eqref{Ybarfinal} can be calculated
from the compressed QR factors \big($\bar{R}_{ij}$ $\forall i,j\in\{1,2,3\}$, see \eqref{eq_Yfbar_ortho}\big).
Using the above Lemma and comparing \eqref{eq_Yf_ortho_data}
and \eqref{eq_Yfbar_ortho}, we get 
\begin{equation}
\bar{\zeta}=\underbrace{\bar{R}_{32}\bar{R}_{22}^{\dagger}}_{\bar{L}_{p}}\bar{W}_{p}\label{eq_zetabar}
\end{equation}

\subsubsection{Projection: SVD Step}

\label{subsection:SVDstep} As mentioned in previous subsection, after computing
the oblique projection $\zeta\in\mathbb{R}^{kp\times N}$ (see \eqref{eqzeta}),
the next step is to perform SVD according to \eqref{eqzetasvd}. Since,
we are interested to compute $\mathcal{S}^{x}$ using only the left singular
vectors of $\zeta$,
we show that an equivalent operation can be performed on $\bar{\zeta}$.
\begin{theorem}
\label{Theorem_thetak} There exists a decomposition of $\bar{\zeta}=\bar{L}_{p}\bar{W}_{p}=\bar{\mathcal{S}}^{x}X\in\mathbb{R}^{kp\times N_{c}}$
such that $\mathcal{R}(\bar{\mathcal{S}}^{x})=\mathcal{R}(\mathcal{S}^{x})$
a.s. where $\bar{\mathcal{S}}^{x}\in\mathbb{R}^{kp\times n}$ and $X\in\mathbb{R}^{n\times N_{c}}$.
\end{theorem}
\begin{proof}
The proof is similar to Theorem 2 in \cite{kedia2023fast} 
\end{proof}
Therefore SVD can be performed on $\bar{\zeta}$ to compute $\mathcal{S}^{x}$ upto similarity transform. This step reduces the computation-cost as we use $\bar{\zeta}\in\mathbb{R}^{kp\times N_{c}}$
to compute $\bar{\mathcal{S}}^{x}$ instead of $\zeta\in\mathbb{R}^{kp\times N}$. The matrix $\mathcal{S}^{u}$ can be easily estimated using \eqref{eq_Yfbar_ortho} 
\begin{equation}\label{eq_Psikbar}
    \begin{bmatrix}\mathcal{S}^{u} & 0\end{bmatrix} = (\bar{R}_{31}-\bar{R}_{32}\bar{R}_{22}^{\dagger}\bar{R}_{21})\bar{R}_{11}^{-1}
\end{equation}
The above Theorem proves that even after data compression we can use appropriate projections to
learn $\mathcal{S}^{x}$ and $\mathcal{S}^{u}$ up to similarity transforms. 

\section{Randomized Iterative LQR (RiLQR)}
The methods described in the previous section allows for efficient computation of the LQR control directly from large quantities of I/O data. This efficiency is essential to learn from the potentially large sized similar system data. However we would like to continuously update the optimal control based on the online data generated in real time from the actual system being controlled. This calls for a iterative method to update $\mathcal{S}^{x}$, $\mathcal{S}^{u}$ and ideally $\mathcal{U}_{i}^{*}(x(i))$ in \eqref{eq_optimalu_transformed}  based on the newly available data at each time instant. We propose a method to partially achieve this objective in this section.

\subsection{Learning control from similar system}

Recall the dynamics of the similar system given in \eqref{ss_similarmodel1} which has previously generated the I/O data set denoted by $\{\tilde{u}(i),\tilde{y}(i)\}_{i=0}^{N_{t}-1}$.
We formulate data matrices from this I/O data, denote them as ${U}_{p(-1)}$, ${U}_{f(-1)}$,
${Y}_{p(-1)}$, ${Y}_{f(-1)}$ and define $H_{-1} := \begin{bmatrix}{U}_{f(-1)}^{T} & {W}_{p(-1)}^{T} & {Y}_{f(-1)}^{T}\end{bmatrix}^{T}\in\mathbb{R}^{2k(m+n)\times N}$. The ${U}_{p(-1)}$ and ${U}_{f(-1)}$ using data from similar system can be represented as:
\begin{equation*}
	{U}_{p(-1)} :=\begin{bsmallmatrix}\tilde{u}(0) & \tilde{u}(1) & \hdots & \tilde{u}(N-1)\\
		\tilde{u}(1) & \tilde{u}(2) & \hdots & \tilde{u}(N)\\
		\vdots & \vdots & \ddots & \vdots\\
		\tilde{u}(k-1) & \tilde{u}(k) & \hdots & \tilde{u}(k+N-2)
	\end{bsmallmatrix}
\end{equation*}
\begin{equation*}
	{U}_{f(-1)} :=\begin{bsmallmatrix}\tilde{u}(k) & \tilde{u}(k+1) & \hdots & \tilde{u}(k+N-1)\\
		\tilde{u}(k+1) & \tilde{u}(k+2) & \hdots & \tilde{u}(k+N)\\
		\vdots & \vdots & \ddots & \vdots\\
		\tilde{u}(2k-1) & \tilde{u}(2k) & \hdots & \tilde{u}(2k+N-2)
	\end{bsmallmatrix}
\end{equation*}
Let us define $\tilde{u}_{p(-1)}$ as the last column of
${U}_{p(-1)}$ so that $\tilde{u}_{p(-1)}= \begin{bsmallmatrix}
\tilde{u}(N-1)\\ \tilde{u}(N) \\ \vdots \\\tilde{u}(k+N-2)
\end{bsmallmatrix}$; and $\tilde{u}_{f(-1)}$ as the last column of ${U}_{f(-1)}$ so that  $\tilde{u}_{f(-1)}= \begin{bsmallmatrix}
\tilde{u}(k+N-1)\\ \tilde{u}(k+N) \\ \vdots \\\tilde{u}(2k+N-2)
\end{bsmallmatrix}$. Similarly, $\tilde{y}_{p(-1)}$
and $\tilde{y}_{f(-1)}$ can be defined. Therefore the last column of $H_{-1}$
is defined as $h_{-1} :=\begin{bsmallmatrix}\tilde{u}_{f(-1)}\\
\tilde{u}_{p(-1)}\\
\tilde{y}_{p(-1)}\\
\tilde{y}_{f(-1)}
\end{bsmallmatrix} \in \mathbb{R}^{2k(m+n)}$.

Following the same notation as in section \ref{section:effiecient_learning_randomized}, the compressed version of $H_{-1}$ is denoted by $\bar{H}_{-1} = H_{-1}\mathcal{C}_{-1}$. Using this $\bar{H}_{-1}$ and our knowledge of the initial state $x(0)$ of the to-be-controlled system \eqref{eq_ssmodel},  the optimal control $\mathcal{U}_{0}^{*}(x(0))$ can be computed using  $\mathcal{S}^{x}$ and  $\mathcal{S}^{u}$ as shown above. The first instance of this control input sequence (say  $u^{*}(0)$) is applied at the starting instant $t=0$ of the controlled operation of system \eqref{eq_ssmodel}. The output $y(0)$ is recorded.

\subsection{Update of H}\label{section:updateh}

Assume that we receive the new data $(u^{*}(0),y(0))$ generated at time $t=0$ from the actual system. This
new information can be used to augment the I/O data matrices as shown below: 
\begin{equation*}
	{U}_{p(0)} :=\begin{bsmallmatrix}\tilde{u}(0) & \tilde{u}(1) & \hdots & \tilde{u}(N-1)&\textcolor{red}{\tilde{u}(N)}\\
		\tilde{u}(1) & \tilde{u}(2) & \hdots & \tilde{u}(N)& \textcolor{red}{\tilde{u}(N+1})\\
		\vdots & \vdots & \ddots & \vdots & \textcolor{red}{\vdots}\\
		\tilde{u}(k-1) & \tilde{u}(k) & \hdots & \tilde{u}(k+N-2)&\textcolor{red}{\tilde{u}(k+N-1)}
	\end{bsmallmatrix}
\end{equation*}
\begin{equation*}
	{U}_{f(0)} :=\begin{bsmallmatrix}\tilde{u}(k) & \tilde{u}(k+1) & \hdots & \tilde{u}(k+N-1)&\textcolor{red}{\tilde{u}(k+N)}\\
		\tilde{u}(k+1) & \tilde{u}(k+2) & \hdots & \tilde{u}(k+N)&\textcolor{red}{\tilde{u}(k+N+1)}\\
		\vdots & \vdots & \ddots & \vdots&\textcolor{red}{\vdots}\\
            \tilde{u}(2k-2) & \tilde{u}(2k-1) & \hdots & \tilde{u}(2k+N-3)&\textcolor{red}{\tilde{u}(2k+N-2)}\\
		\tilde{u}(2k-1) & \tilde{u}(2k) & \hdots & \tilde{u}(2k+N-2)&\textcolor{blue}{{u}^{*}(0)}
	\end{bsmallmatrix}
\end{equation*}
Similarly $Y_{p(0)}$ and $Y_{f(0)}$ can also be augmented. 
So, the updated new column to be added to $H_{-1}$ can be written as $h_{0}=\begin{bmatrix}u_{f0}\\
u_{p0}\\
y_{p0}\\
y_{f0}
\end{bmatrix}$ where $u_{f0}$ denotes the last column of $U_{f(0)}$ and so on.
Therefore, $H_0 = \begin{bmatrix}
    H_{-1} & \textcolor{blue}{h_{0}}
\end{bmatrix}$. Now following the same notation as in section \ref{section:effiecient_learning_randomized}, we perform data compression on $H_{0}$ leading to:
\begin{equation}\begin{aligned}\label{eq_barH_update}
    \bar{H}_{0} &= H_{0}\mathcal{C}_{0}= \begin{bmatrix}
    H_{-1} & \textcolor{blue}{h_{0}}
\end{bmatrix}\begin{bmatrix}
    \mathcal{C}_{-1}\\ \textcolor{purple}{c_{0}^{T}}\end{bmatrix}\\
    &= H_{-1}\mathcal{C}_{-1} + \textcolor{blue}{h_{0}}\textcolor{purple}{c_{0}^{T}}\\
    &=\bar{H}_{-1} + \textcolor{blue}{h_{0}}\textcolor{purple}{c_{0}^{T}}
\end{aligned}
\end{equation}
From time instant $t=1$  onwards the compressed data matrix at the $t^{th}$ instant can be obtained from the data matrix at the $(t-1)^{th}$ instant through an iteration identical to \eqref{eq_barH_update}  above. 
\begin{equation}\label{eq_barH_update1}
    \bar{H}_{t}=\bar{H}_{t-1} + h_{t}c_{t}^{T}
\end{equation}
where $h_t$ is the last column of $H_t$ augmented with the current data point $(u^{*}(t),y(t))$ as shown above, and $c_{t}$ is a iid Gaussian vector of appropriate size. Therefore the new information can be included by iterating with the previously compressed data matrix. Note that the size of the compressed matrix remains invariant even after addition of infinite data points. Though not pursued in this article, one may weigh new information differently than previously collected data by weighing the rank one update by a weighing factor $\gamma>0$ as in 
\begin{equation}\label{eq_barH_update1}
    \bar{H}_{t}=\bar{H}_{t-1} + \gamma h_{t}c_{t}^{T}.
\end{equation}

\subsection{Update on QR decomposition: Iterative rank-1 update}

 Once we have computed $\bar{H}_{t}$, the next step would be to perform QR decomposition ($\bar{H}_{t}^{T}$) followed by SVD and the subsequent steps described in the previous section. In this section, we show that the QR decomposition can be iterated directly without computing $\bar{H}_{t}$ explicitly. Clearly \eqref{eq_barH_update} is a rank-1 update on $\bar{H}_{t-1} \in \mathbb{R}^{s \times N_{c}}$ where $s:= 2k(m+n)$.  Therefore we can directly iterate on the $Q$ and $R$ factors.
 
Let the QR decomposition of $\bar{H}_{t-1}^{T} = Q_{t-1}R_{t-1}$ 
and $\bar{H}_{t} = \bar{H}_{t-1}^{T}+h_{t}c_{t}^{T}=Q_{t-1}(R_{t-1}+z_{t}c_{t}^{T})$
where $z_{t}=Q_{t-1}^{T}h_{t}$. We further define $L_{t}^{J}:=J_{1(i)}^{T}J_{2(i)}^{T}\hdots J_{(s-1)(i)}^{T}$ where $J_{(q)(i)}$ $\forall q \in  \{1, 2, \hdots, s-1\}$ are rotation matrices in planes $q$ and $q+1$ such that $L_{t}^{J}z_{t} = \pm ||z_{t}||_{2}e_{1}$. These rotation matrices can be easily computed using Algorithms 5.2.4 in \cite{golub2013matrix} and requires $\mathcal{O}(s^2)$ computation. It is easy to see that $M := L_{t}R_{t-1}$ is upper Hessenberg. Consequently, $L_{t}^{J}(R_{t-1} + z_{t}c_{t}^{T}) = M_{t-1} \pm ||z_{t}||_{2}e_{1}c_{t}^{T}$ is also upper Hessenberg. Next we define, $L_{t}^{G}:=G_{(s-1)(i)}^{T}G_{(s-2)(i)}^{T}\hdots G_{1(i)}^{T}$ where $G_{(q)(i)}$ $\forall q \in  \{1, 2, \hdots, s-1\}$ are rotation matrices such that $L_{t}^{G}(M_{t-1} \pm ||z_{t}||_{2}e_{1}c_{t}^{T})$ is upper triangular. Therefore the
iterative update on QR decomposition is given by: 
\begin{equation}\label{eq_qr_update}
\begin{aligned} & Q_{t}=Q_{t-1}Z_{t}\\
 & R_{t}=Z_{t}^{T}R_{t-1}+\Delta R_{t}
\end{aligned}
\end{equation}
where $Z_{t} = L_{t}^{G}L_{t}^{J}$ and $\Delta R_{t}=L_{t}^{G}||z_{t}||_{2}e_{1}c_{t}^{T}$. 

It is well known \cite{golub2013matrix} that the above rank-1 QR update requires $\mathcal{O}(s^2)$ computation. Hence, the above algorithm requires $\mathcal{O}(s^{2})$ computation at each iteration as
compared to $\mathcal{O}(s^{3}+s^2l)$ in case of repeated QR decomposition on $\bar{H}_{t}$. It is important to note that if we had not compressed $H_{t}$, then the required QR on $H_{t}^{T}$ would have required $\mathcal{O}(s^2N)$ computation at each instant making it prohibitively large due to $N >> s$. Moreover, in such an uncompressed scenario, the data size ($N$ ) would grow with time making this computation even more intractable. Therefore, there is significant saving in computation in this step due to iterative compression and the rank-1 update of the compressed QR factors. 

The proposed pseudo-iterative control scheme (RiLQR) is summarized in two algorithms: Algorithm \ref{algo:initial_similar_system} describes the process of learning the initial optimal control by leveraging data from a similar system, while Algorithm \ref{algo:RiLQR}  iteratively learns the modified control by augmenting old information from the similar system with online data from the actual system.
 
\begin{algorithm}[h]
\label{algo:initial_similar_system} \textbf{Input:} Load input-output data from
similar system.\\
 Formulate data matrices $U_{p(-1)}$, $U_{f(-1)}$, $Y_{p(-1)}$ and $Y_{f(-1)}$
from input-output data. Stack them to form $H_{-1}=\begin{bmatrix}U_{f(-1)}^{T} & U_{p(-1)}^{T} & Y_{p(-1)}^{T} & Y_{f(-1)}^{T}\end{bmatrix}^{T}$
matrix.\\
Store the last column of $H_{-1}$ as $h_{-1}=H_{-1}(:,end)$.
\\
 Generate random Gaussian iid matrix $\mathcal{C}_{-1}\in\mathbb{R}^{N\times N_{c}}$.\\
 Perform randomized data compression on $H_{-1}$ using $\mathcal{C}_{-1}$ defined as $\bar{H}_{-1}:=H_{-1}\mathcal{C}_{-1}$.\\
 Perform QR decomposition on $\bar{H}_{-1}^{T}=Q_{-1}R_{-1}$.\\
 Extract $\bar{L}_{p}$ and $\mathcal{S}_{-1}^{u}$ from $R_{-1}$
obtained above using \eqref{eq_zetabar} and \eqref{eq_Psikbar} respectively.
\\
 Perform SVD using $\bar{L}_{p}$ and $\bar{W}_{p}$ to
estimate $\bar{\mathcal{S}}_{-1}^{x}$ \eqref{eqzetasvd}.\\
 Compute $\mathcal{U}_{-1}^{*}$ using \eqref{eq_optimalu_transformed}. Store
first input sequence from $\mathcal{U}_{-1}^{*}$ as $u^{*}(0)$.\\
 \textbf{Output:} $\mathcal{S}_{-1}^{x}$, $\mathcal{S}_{-1}^{u}$,
$Q_{-1}$, $R_{-1}$, $h_{-1}$ and $u^{*}(0)$. \caption{Learning control from similar system}
\end{algorithm}

\begin{algorithm}[h]
\label{algo:RiLQR} Initialize: $t\gets0$\\
 \While{$t\neq{T_{iter}}$}{ \textbf{Input:} $\mathcal{S}_{t-1}^{x}$,
$\mathcal{S}_{t-1}^{u}$, $Q_{t-1}$, $R_{t-1}$, $h_{t-1}$ and $u^{*}(t)$.\\
 Generate output data $y(t)$ using $u^{*}(t)$ from actual system
$(A,B)$. Record the I/O pair as $(u^{*}(t),y(t))$.\\
 Update $h_{t}$ based on $ h_{t-1}$ using current measurement from actual
system $(u^{*}(t),y(t))$ (see section \ref{section:updateh}).\\
 Generate random Gaussian vector $c_{t}$.\\
 Iteration on QR decomposition to obtain $R_{t}$ using $Q_{t-1}$,
$R_{t-1}$, $h_{t}$ and $c_{t}$ using \eqref{eq_qr_update}. \\
  From $R_{t}$ obtained above compute $\mathcal{S}_{t}^{x}$ and
$\mathcal{S}_{t}^{u}$ as performed in Algorithm \ref{algo:initial_similar_system}.\\
 Compute $\mathcal{U}_{t}^{*}$ based on $\mathcal{S}_{t}^{x}$
and $\mathcal{S}_{t}^{u}$ using \eqref{eq_optimalu_transformed}. Store first sequence from $\mathcal{U}_{t}^{*}$
denoted by $u^{*}(t+1)$\\
 \textbf{Output:}. $\mathcal{S}_{t}^{x}$, $\mathcal{S}_{t}^{u}$,
$Q_{t}$, $R_{t}$, $h_{t}$ and $u^{*}(t+1)$.\\
 $t\gets t+1$\;} \caption{Learning control from similar system and actual system }
\end{algorithm}

\subsection{Efficiency of the proposed algorithm}\label{section:algo_efficiency_RiLQR}
The computation complexity of the proposed algorithm is presented in Table \ref{tab:computation_complexity}. The main computation burden of the proposed algorithm is due to initial data compression step i.e. matrix multiplication of $H_{-1}$ and $\mathcal{C}_{-1}$ (see step 5 in Algorithm \ref{algo:initial_similar_system}) that requires $\mathcal{O}(k^2N)$ computation. There are two ways to address this issue:
\begin{enumerate}
    \item Offline computation of $\bar{H}_{-1}$: Since the data was generated from a similar system, hence the compression can be done before the start of real time operation on the actual system.
    \item Matrix multiplication is a highly parallel operation hence can be implemented efficiently.
\end{enumerate}
Except the initial data compression step, all other steps of the proposed algorithm is $\mathcal{O}(k^2N_{c}) \approx \mathcal{O}(k^3)$.
\begin{table}[h]
\begin{center}
\caption{Computation complexity of RiLQR}
\label{tab:computation_complexity}
\begin{tabular}{|c|c|c|}
\hline
\textbf{}  & \textbf{Steps}                                                                                               & \textbf{Proposed}   \\ \hline
\textbf{1} & \begin{tabular}[c]{@{}c@{}}Initial data compresssion\\ $\bar{H}_{-1} = H_{-1} \mathcal{C}_{-1}$\end{tabular} & $\mathcal{O}(k^2N)$ \\ \hline
\textbf{2} & \begin{tabular}[c]{@{}c@{}}Initial QR decomposition\\ $\bar{H}_{-1} = Q_{-1} R_{-1}$\end{tabular} & $\mathcal{O}(k^3)$ \\ \hline
\textbf{3} & Rank-1 QR update                                                                                             & $\mathcal{O}(k^2)$  \\ \hline
\textbf{4} & $\mathcal{S}^{x}$                                                                                                         & $\mathcal{O}(k^3)$  \\ \hline
\textbf{5} & $\mathcal{S}^{u}$                                                                                            & $\mathcal{O}(k^3)$  \\ \hline
\textbf{6} & $\mathcal{U}^{*}$                                                                                            & $\mathcal{O}(k^3)$  \\ \hline
\end{tabular}
\end{center}
\end{table}

\section{Numerical Example}

We illustrate the advantage of using historical data set obtained from a similar system in the proposed method, over the standard exploration-exploitation approach, where the system is unknown. Consider a second order, discrete-time system whose matrices are given by
\begin{equation*}\begin{aligned}
    A & =\begin{bmatrix} 1.0 & 0.40\\
0.005 & -0.99\\
\end{bmatrix}; B=\begin{bmatrix}0.2\\
0.5\\
\end{bmatrix};
C=I_{2}
\end{aligned}
\end{equation*}

Firstly, as part of the exploratory phase in the standard approach, a Gaussian iid white noise sequence of variance $\sigma^{2}_{u_E}$ is used to drive the system forward in open loop through $T_{explore}= 50$ time-steps, and the generated I/O data set $\{\tilde{u}(i),\tilde{y}(i)\}_{i=0}^{50}$ is used to identify the system using the Multivariable Subspace Identification: MOESP \cite{MOESP_Verhaegan} algorithm. This algorithm requires the creation of the Hankel matrices (see \eqref{hankel_matrix}) The width of these matrices depend on the length of the collected data. Since we would like to keep the exploration phase to be as short as possible to keep down the LQR cost, we choose $T_{explore}= 50$ to make the data matrix $H$ approximately square (see \eqref{eq_Hmatrix}).  
The variance of the noise sequence  $e(t)$ is kept around $0.01$ in all the experiments. The model based LQR controller uses the estimated system matrices to compute the optimal feedback gain  \cite{borrelli2017predictive} with a prediction horizon $k_p=4$  and is run in closed loop for $T_{exploit}= 150$ time-steps as part of the exploitation phase. The  $T_{exploit}= 150$ is kept sufficiently long for the system to become stable.  The total run time is: $T_{iter} = T_{explore} + T_{exploit}$. In order to apply the proposed methodology, we set a parallel
open-loop experiment to collect the historical data with a similar system with perturbed eigenvalues, given as
\begin{equation*}
\tilde{A} =\begin{bmatrix} 0.80 & 0.30\\
0.105 & -0.89\end{bmatrix}\hspace{0.3cm}
\tilde{B}=\begin{bmatrix}0.21\\
0.6\\
\end{bmatrix}
\end{equation*}

The other system matrices are kept the same as in the original system. The eigenvalues of the original system are $1.001$ and $-0.991$, and of the similar system are  are $0.818$ and $-0.908$. For the similar system experiment, without loss of any generality, the input signal is chosen to be a white Gaussian noise sequence with zero mean, covariance $\sigma^{2}_{ud}  = 1 $ with data length $N_{t} =$100,000. The I/O data is collected and stored for future use. 
Algorithm \ref{algo:initial_similar_system}  is applied on the similar system data collected and then the actual system $(A,B)$ is run in closed loop in conjunction with Algorithm \ref{algo:RiLQR} for the same $T_{iter}= 200$, where the RiLQR controller learns from the actual system outputs in real time and generates the optimal input sequence. The state and input penalty matrices are chosen as $Q = P = I_{2}$ and $R=1$ respectively.

For the proposed algorithm, the total optimal cost ($J_{RiLQR}$) is computed using (3) over $T_{iter}$ time-steps. The total cost for the model based LQR method (mLQR) is summed over the exploration and exploitation phases: $J_{mLQR} = J_{explore} + J_{exploit}$. 

Clearly the model based control cost $J_{explore}$ is dependent on both the duration of the explore phase as well as the variance of the input sequence used during the explore phase. The choice of  $T_{explore}= 50$  is explained above. The variance of the input noise during the explore phase, denoted by $\sigma^{2}_{u_E}$, also affects the accuracy of the estimated system parameters and hence the control. For very low  $\sigma^{2}_{u_E}$, the system is misidentified by the MOESP algorithm owing to low excitation; which is reflected in the worsening control performance for very low values of  $\sigma^{2}_{u_E}$. Hence a very low $\sigma^{2}_{u_E}$ may decrease the $J_{explore}$ but will increase the $J_{exploit}$. We repeat the model based experiment described above for several values of $\sigma^{2}_{u_E}$. It can be seen from Table \ref{tab:template}, that the variation of $J_{mLQR}$ vs  $\sigma^{2}_{u_E}$ is as expected. It is also clear that even the minimum value of  $J_{mLQR}$ ($=55.27$) achieved approximately with  $\sigma^{2}_{u_E}= 3.2\text{e-}5$, is far greater than that resulting from the proposed method. The input and the state trajectories for the model based control vs the proposed algorithm are shown  for two different values of  $\sigma^{2}_{u_E}$:  $\sigma^{2}_{u_E}= 0.83$ for Figure \ref{fig:systemresponse1} and  $\sigma^{2}_{u_E}= 9.5\text{e-}8$ for Figure \ref{fig:systemresponse2}.

\begin{table}[h]
\begin{center}
\caption{Cost and eigenvalue Variation with respect to input variance}
\label{tab:template}
\scalebox{0.85}{
\begin{tabular}{|c|c|c|c|c|c|c|}
\hline
\textbf{Cases} & \bm{$\sigma^{2}_{u_E}$} & \textbf{SNR} & \bm{$J_{explore}$} & \bm{$J_{exploit}$} & \bm{$J_{mLQR}$} & \bm{$J_{RiLQR}$} \\ \hline
\textbf{1}     &   1.15 & 98.90 &    625.16 & 20.0 & 645.16 & 5.62\\ \hline
\textbf{2}     &  0.83& 79.95&    232.18&16.17&248.36&6.25\\
\hline
\textbf{3}     &   0.08& 9.21&    79.83&3.15&82.98&4.72\\
\hline
\textbf{4}     &    0.001& 0.08&    53.05&2.92&55.97&5.69\\
\hline
\textbf{5}     &  3.2e-5& 3.7e-3&    53.25&2.02&55.27&5.40\\
\hline
\textbf{6}     &  9.5e-8& 8.4e-6&    28.60&55.68&84.29&5.05\\
\hline
\textbf{7}     &   1.1e-10& 1.1e-8&    49.89&72.96&122.86&5.71\\
\hline
\end{tabular}}
\end{center}
\end{table}


\begin{figure}[h]
\centering
\includegraphics[width=8cm]{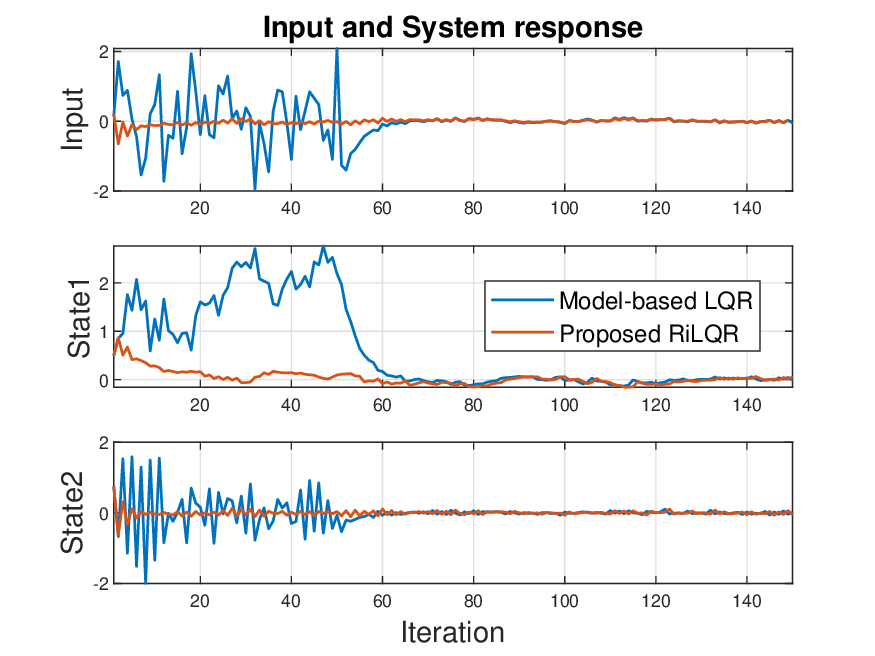}
\caption{State response comparison for mLQR and RiLQR for\\ $T_{iter}$ = 150 with SNR  = 79.95}
\label{fig:systemresponse1}
\end{figure}

\begin{figure}[h]
\centering
\includegraphics[width=8cm]{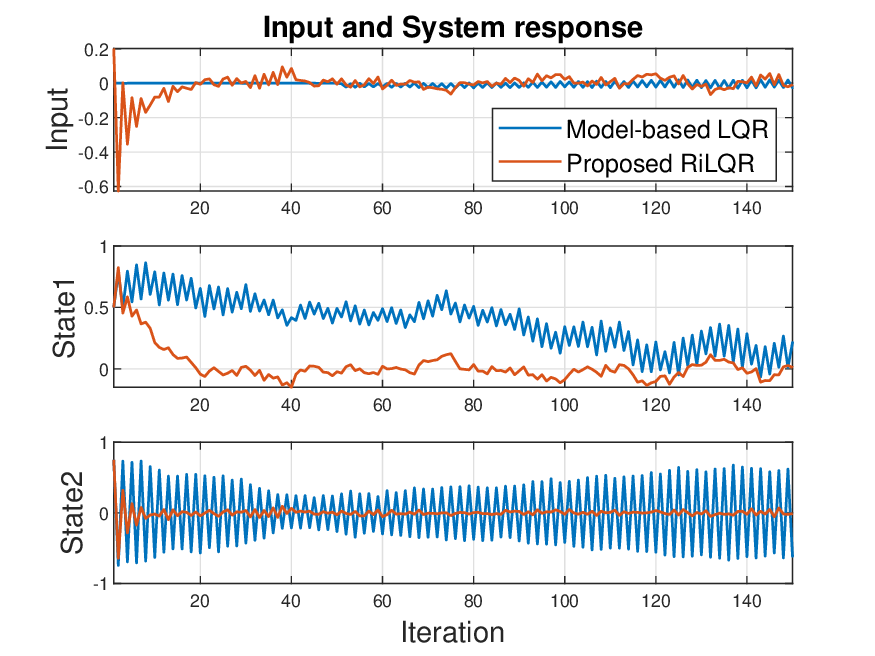}
\caption{State response comparison for mLQR and RiLQR for \\$T_{iter}$ = 150 with SNR = 8.4e-6}
\label{fig:systemresponse2}
\end{figure}

\section{CONCLUSIONS}

A data driven iterative LQR which seamlessly combine past and online data through randomized compression, is presented in this paper. While not included in this preliminary study, it is expected that discrepancies with the historical model and/or rapid real time variations in model parameters will adversely affect the efficacy of the computed control sequence. Convergence proofs and performance guarantees will require restrictions on the admissible model variations, while constraints on the states/inputs might be challenging to handle in a truly iterative framework. The performance of the proposed algorithm vs more sophisticated exploitation-exploration scheme should be investigated in the future.

\addtolength{\textheight}{-10cm} 
\bibliographystyle{ieeetr}
\bibliography{ifacconf}

\end{document}